\documentclass[10pt,twocolumn,twoside]{IEEEtran}

\IEEEoverridecommandlockouts                              	

\usepackage{amssymb,amsmath,amsfonts,amsthm}
\usepackage{algorithm,algorithmic}
\usepackage{cite}
\usepackage{float}
\usepackage{epsfig}
\usepackage{graphicx}
\usepackage{graphics}
\usepackage{subfigure}
\usepackage{url}
\usepackage{xspace}
\usepackage[usenames,dvipsnames]{color}
\usepackage{soul}
\usepackage{mathtools}

\floatstyle{ruled}
\newfloat{model}{H}{mod}
\floatname{model}{\footnotesize Model}
\newfloat{notatio}{H}{not}
\floatname{notatio}{\footnotesize Notation}

\newenvironment{varalgorithm}[1]
  {\algorithm\renewcommand{\thealgorithm}{#1}}
  {\endalgorithm}

\newenvironment{list4}{
	\begin{list}{$\bullet$}{%
			\setlength{\itemsep}{0.05cm}
			\setlength{\labelsep}{0.2cm}
			\setlength{\labelwidth}{0.3cm}
			\setlength{\parsep}{0in} 
			\setlength{\parskip}{0in}
			\setlength{\topsep}{0in} 
			\setlength{\partopsep}{0in}
			\setlength{\leftmargin}{0.16in}}}
	{\end{list}}

\usepackage{url,changebar,bm,xspace,dsfont}
\let\mathbb=\mathds 
\def\day{\mathrm{day}} 

\newtheorem{theorem}{Theorem}

\newtheorem{defn}{Definition}

\newtheorem{assum}{Assumption}

\newtheorem{remark}{Remark}

\ifCLASSINFOpdf
 \else
\fi

\usepackage{amsmath,amsfonts,amssymb,amscd}
\usepackage{capt-of}
\usepackage{color}
\usepackage{cite}
\usepackage{verbatim}
\usepackage{hyperref}

\hypersetup{
    colorlinks = true,
    linkcolor = [rgb]{0,0,1},
    anchorcolor = [rgb]{0,0,1},
    citecolor = [rgb]{0.9,0.5,0},
    filecolor = [rgb]{0,0,1},
    urlcolor = [rgb]{0.9,0.5,0},
    bookmarksopen = true,
    bookmarksnumbered = true,
    breaklinks = true,
    linktocpage,
    colorlinks = true,
    linkcolor = [rgb]{0.2,0.6,0.2},
    urlcolor  = [rgb]{0,0,1},
    citecolor = [rgb]{0.9,0.5,0},
    anchorcolor = [rgb]{0.2,0.6,0.2},
}

\newtheorem{lemma}{\bfseries Lemma}

\begin{document}

\title{\LARGE \bf Optimal Database Allocation in Finite Time with Efficient Communication and Transmission Stopping over Dynamic Networks}

\author{Apostolos~I.~Rikos, Christoforos N. Hadjicostis, and Karl~H.~Johansson
\thanks{Apostolos~I.~Rikos and  K.~H.~Johansson are with the Division of Decision and Control Systems, KTH Royal Institute of Technology, SE-100 44 Stockholm, Sweden. They are also affiliated with Digital Futures, SE-100 44 Stockholm, Sweden. E-mails: {\tt \{rikos,kallej\}@kth.se}.}
\thanks{C. N. Hadjicostis is with the Department of Electrical and Computer Engineering, University of Cyprus, 1678 Nicosia, Cyprus.  E-mail:{\tt~chadjic@ucy.ac.cy}.}
\thanks{This work was supported by the Knut and Alice Wallenberg Foundation and the Swedish Research Council.}
}

\maketitle
\thispagestyle{empty}
\pagestyle{empty}

%
%
%
%
\begin{abstract}
In this paper, we focus on the problem of data sharing over a wireless computer network (i.e., a wireless grid). 
Given a set of available data, we present a distributed algorithm which operates over a dynamically changing network, and allows each node to calculate the optimal allocation of data in a finite number of time steps. 
We show that our proposed algorithm (i) converges to the optimal solution in finite time with very high probability, and (ii) once the optimal solution is reached, each node is able to cease transmissions without needing knowledge of a global parameter such as the network diameter. 
Furthermore, our algorithm (i) operates exclusively with quantized values (i.e., each node processes and transmits quantized information), (ii) relies on event-driven updates, and (iii) calculates the optimal solution in the form of a quantized fraction which avoids errors due to quantization. 
Finally, we demonstrate the operation, performance, and potential advantages of our algorithm over random dynamic networks. 
\end{abstract}


%
%
%
%
\section{Introduction}\label{sec:intro}



Wireless computer networks (or wireless grids) comprise of different electronic devices (or nodes), which share their resources with other devices in a distributed manner. 
Various users or devices may request access to one device's stored data. 
In order to reduce the average waiting time for users or devices who request access, data allocation is the procedure of allocating the available data to different nodes according to their available memory capacity such that specific performance objectives are achieved. 

In wireless computer networks, data comprise an important resource that needs to be managed efficiently. 
Optimal allocation of data can heavily influence the operational performance of the network \cite{2009:Wang_Jea}. 
In general, resource allocation can be formulated as an optimization problem but solving the optimal allocation problem over dynamic networks with quantized communication is challenging due to the heterogeneity of the network and the nonlinear nature of the communication constraints. 
Centralized solutions consider gathering the available data to a central scheduler; however, these solutions are not ideal as they lack scalability and they impose heavy computational and storage requirements on the central scheduler. 
For this reason, there has been interest towards distributed algorithms that solve the optimal allocation problem \cite{2020:Themis_Kalyvianaki, 2020:Doostmohammadian_Charalambous, 2021:Rikos_TaskScheduling, 2015:Garcia_Hadjicostis}. 

Distributed optimization has received significant attention recently due to its wide variety of applications \cite{2004:Rabbat, 2012:TsianosLawlorRabbat, 2019:Yang_etal, boyd2006optimal, nedic2018improved, 2010:Johansson, 2010:Nedic, Rabbat:2018IEEEProceedings, 2015:Cherukuri_Cortes}. 
Most works in the current literature assume that network devices process and exchange real values and are able to reach asymptotic convergence within some error \cite{2020:Themis_Kalyvianaki, 2015:Garcia_Hadjicostis}. 
In practical applications of wireless computer networks, devices need to exchange information messages with finite length (i.e., quantized messages) which allows for a more efficient usage of the available network resources (e.g., energy, processing power, etc.). 
Also, they need to operate over networks which may be dynamic due to changes over the sensing radius of the various devices \cite{2020:Doostmohammadian_Charalambous, 2021:Rikos_TaskScheduling}. 
Additionally, in order to preserve available energy resources, it is desirable for devices to converge in finite time and to stop transmitting once convergence has been achieved \cite{2020:Themis_Kalyvianaki, 2021:Rikos_TaskScheduling}. 
In this paper we propose an algorithm that combines all previously mentioned characteristics: it solves the optimal allocation problem over dynamic networks in finite time, while exhibiting transmission stopping capabilities. 

\vspace{.2cm}

\noindent
\textbf{Main Contributions.}
We focus on the problem of data sharing over a wireless computer network (i.e., a wireless grid). 
We aim to balance the data storage between nodes by distributively allocating the available data per available memory in the network. 
We consider the realistic scenario where nodes process and exchange quantized information.  
We also consider that various changes of each node's sensing radius result in dynamically changing connections in the communication network. 
Our algorithm is analyzed for optimal data allocation over a wireless computer network. 
In this scenario, we want to reduce the average waiting time for users or devices who request access to the stored data. 
However, please note that the proposed algorithm could be adopted in a wide variety of other related applications. 
The main contributions of the paper are the following. 
\begin{itemize}
    \item We present a distributed algorithm which solves the optimal data allocation problem over a dynamic network; see Algorithm~\ref{algorithm_max}. 
    \item We show that our algorithm converges in a finite number of time steps, and each node is able to calculate the exact solution without introducing errors due to quantization; see Section~\ref{ConvDynamicAlg}.   
    \item Once our algorithm converges to the optimal solution, each node ceases transmissions without needing knowledge of a global network parameter (e.g., the diameter of the network). 
    Note that this allows our algorithm to require no reinitialization when there is a change over the network (e.g., when a node enters or leaves the network or when the diameter changes); see Section~\ref{dynamic_oper_analysis}. 
    \item We analyze the convergence time of the algorithm and show that it relies on the time-varying connectivity (which is determined by the time needed to communicate among pairs of nodes), rather than the size of the network; see Theorem~\ref{converge_Alg1}. 
    \item We present simulations of our algorithm where we show its finite time convergence to the exact solution and its transmission stopping capabilities; see Section~\ref{sec:results}. 
\end{itemize}
The operation of our proposed algorithm relies on each node's ability to (i) directly transmit a set of values, (ii) broadcast its state, and (iii) remember the sets of nodes which have received the broadcast transmissions. 
Initially, each node broadcasts its state, and also remembers the set of nodes which received the initial broadcast. 
If the dynamic network changes and a new neighboring node appears (which did not receive the current state), then the corresponding node broadcasts again its state and remembers also the new neighboring node. 
However, if a node's state changes, (i) it forgets the set of nodes which received the previous state values, (ii) broadcast its new state, and (iii) remembers only nodes which have received the new state. 
This allows our proposed algorithm to exhibit finite time convergence and transmission stopping when operating over dynamically changing communication networks. 

Unlike our work in this paper, most of the current literature comprises of algorithms which operate with real values and converge asymptotically within some error \cite{2020:Themis_Kalyvianaki, 2015:Garcia_Hadjicostis}. 
Our paper, along with \cite{2021:Rikos_TaskScheduling, 2020:Doostmohammadian_Charalambous}, aims to pave the way for finite time algorithms, which operate solely with quantized values to address resource allocation problems. 
To the authors knowledge, the proposed algorithm is the first algorithm in the current literature which guarantees finite time convergence and transmission stopping for the case where the underlying network is dynamic without needing knowledge of a global parameter such as the diameter of the network (e.g., see \cite{2021:Rikos_TaskScheduling}).

%
%
%
%
\section{NOTATION AND BACKGROUND}\label{sec:preliminaries}

The sets of real, rational, integer and natural numbers are denoted by $ \mathbb{R}, \mathbb{Q}, \mathbb{Z}$ and $\mathbb{N}$, respectively. 
The symbol $\mathbb{Z}_+$ denotes the set of nonnegative integers. 
For any real number $a \in \mathbb{R}$, the floor $\lfloor a \rfloor$ denotes the greatest integer less than or equal to $a$ while the ceiling $\lceil a \rceil$ denotes the least integer greater than or equal to $a$. 


\vspace{.2cm}

\noindent
\textbf{Graph-Theoretic Notions.}
Consider a \textit{dynamic} network of $n$ ($n \geq 2$) nodes communicating only with their immediate neighbors. 
The communication topology can be captured by a dynamic undirected graph, called \textit{dynamic communication graph}. 
A dynamic graph is defined as $\mathcal{G}[k] = (\mathcal{V}, \mathcal{E}[k])$, where $\mathcal{V} =  \{v_1, v_2, \dots, v_n\}$ is the set of nodes and $\mathcal{E}[k] \subseteq \mathcal{V} \times \mathcal{V} - \{ (v_j, v_j) \ | \ v_j \in \mathcal{V} \}$ is the set of edges (self-edges excluded). 
An edge from node $v_i$ to node $v_j$ is denoted by $m_{ji} \triangleq (v_j, v_i) \in \mathcal{E}[k]$, and captures the fact that node $v_j$ and node $v_i$ can exchange information ($v_j$ can transmit to $v_i$ and $v_i$ can transmit to $v_j$) at time step $k$. 
Note here that if $(v_j, v_i) \in \mathcal{E}[k]$ then $(v_i, v_j) \in \mathcal{E}[k]$. 
At time step $k$, the subset of nodes that can directly transmit information to node $v_j$ is called the set of neighbors of $v_j$ and is represented by $\mathcal{N}_j[k] = \{ v_i \in \mathcal{V} \; | \; (v_j,v_i)\in \mathcal{E}[k]\}$. 
The cardinality of $\mathcal{N}_j[k]$ at time step $k$, is called the \textit{degree} of $v_j$ and is denoted by $\mathcal{D}_j[k] = | \mathcal{N}_j^-[k] |$. 
Given a collection of graphs $\mathcal{G}[k] = (\mathcal{V}, \mathcal{E}[k])$ for $k = 1, 2, ..., m$, where $m \in \mathbb{N}$, the \textit{union graph} is defined as $\mathcal{G}^{1, 2, ..., m}_d = (\mathcal{V}, \cup_{k = 1}^{m} \mathcal{E}[k])$. 
A collection of graphs is said to be \textit{jointly connected}, if its corresponding union graph $\mathcal{G}^{1, 2, ..., m}$ forms a connected graph (i.e., for each pair $v_j, v_i \in \mathcal{V}$, $v_j \neq v_i$, there exists a \textit{path}\footnote{A \textit{path} from $v_i$ to $v_j$ exists between time steps $k, k+1, ..., k + \tau$ if we can find a sequence of vertices $v_i \equiv v_{l_0},v_{l_1}, \dots, v_{l_t} \equiv v_j$ such that $(v_{l_{\tau+1}},v_{l_{\tau}}) \in \mathcal{E}[k + \tau]$ for $ \tau = 0, 1, \dots , t-1$.} from $v_i$ to $v_j$).


\vspace{.2cm}

\noindent
\textbf{Node Operation.}
The operation of each node $v_j \in \mathcal{V}$ respects the quantization of information flow.  
At time step $k \in \mathbb{Z}_+$ (where $\mathbb{Z}_+$ is the set of nonnegative integers), each node $v_j$ maintains the mass variables $y_j[k] \in \mathbb{Z}$ and $z_j[k] \in \mathbb{Z}_+$, which are used to communicate with other nodes by either transmitting or receiving messages; the state variables $y^s_j[k] \in \mathbb{Z}$, $z^s_j[k] \in \mathbb{Z}_+$ and $q_j^s[k] = \frac{y_j^s[k]}{z_j^s[k]}$, which are used to store the received messages and calculate result of the optimization operation; the transmission variables $S\_br_j \in \mathbb{N}$ and $M\_tr_j \in \mathbb{N}$, which are used to decide whether $v_j$ will broadcast its state variables or transmit its mass variables via a direct transmission. 

For the case where each node $v_j$ is required to perform a direct transmission, we assume that $v_j$ is aware of its out-neighbors and can directly transmit messages to each out-neighbor separately. 
In the proposed distributed algorithm, in order to randomly determine which out-neighbor to transmit to, each node $v_j$ assigns a nonzero probability $b_{lj}[k]$ to each of its edges $m_{lj}$ where $v_l \in \mathcal{N}_j[k]$ (note that there is always a virtual self-edge which means that the probability assigned to the self-edge is nonzero). 
For every node, this probability assignment can be captured by an $n \times n$ column stochastic matrix $\mathcal{B}[k] = [b_{lj}[k]]$. 
A simple choice is to set these probabilities to be equal, i.e.,
\begin{align*}
b_{lj}[k] = \left\{ \begin{array}{ll}
         \frac{1}{\mathcal{D}_j^+[k]+1}, & \mbox{if $l = j$ or $v_{l} \in \mathcal{N}_j[k]$,}\\
         0, & \mbox{otherwise.}\end{array} \right. 
\end{align*}
Each nonzero entry $b_{lj}[k]$ of matrix $\mathcal{B}[k]$ represents the probability of node $v_j$ transmitting towards out-neighbor $v_l \in \mathcal{N}_j[k]$ through the edge $m_{lj}$. 


\vspace{.2cm}

\noindent
\textbf{Modelling of Wireless Grid and Database.}
A wireless grid is modeled as a set of $\mathcal{V}$ nodes (or wireless sensors) and each node is denoted as $v_{i} \in \mathcal{V}$. 
In most data grids, all participating nodes are interconnected with undirected communication links. 
Furthermore, various changes over the sensing range of each node impose a dynamic nature to the network topology. 
This means that the network topology forms a dynamic undirected graph. 

For modelling databases, we borrow notation from \cite{2021:Rikos_TaskScheduling}. 
Specifically, the database to be allocated in the network is $D_{dat}$. 
The data sets which consists the database are $d_j \in D_{dat}$ (where $j \in \{1,\ldots, |D_{dat}| \}$). 
The required memory for each data set $d_j$ to be stored (which is known before the optimization operation) is $\mu_j$.
Thus, the total required memory for database $D_{dat}$ is $\mu \coloneqq \sum_{v_{i}\in \mathcal{V}} \mu_j$. 
The total load of data at each node $v_j$, due to incoming data in the network is $l_j$. 
The time period for which the optimization operation is executed (before the next optimization operation) is $T_{o}$. 
The total memory of node $v_j$ is $\nu_j^{\max}$. 
The total memory in the network is $\nu^{\max} \coloneqq \sum_{v_{j}\in \mathcal{V}} \nu_j^{\max}$. 
The amount of unavailable memory of node $v_j$ due to previously stored data is $\delta_j[m]$. 
The total amount of unavailable memory in the network $\delta_{\mathrm{tot}}[m] = \sum_{v_{j}\in \mathcal{V}} \delta_{j}[m]$. 
The amount of available memory of node $v_j$ at optimization step $m$ (i.e., at time step $m T_{o}$) is $\nu_j^{\mathrm{avail}}[m] \coloneqq \nu_j^{\max} - \delta_{j}[m]$. 
The total amount of available memory in the network is $\nu^{\mathrm{avail}}[m] \coloneqq \sum_{v_{j}\in \mathcal{V}} \nu_j^{\mathrm{avail}}[m]$. 

\section{Problem Formulation}\label{sec:probForm}

Let us consider a wireless computer network modeled as a dynamic graph $\mathcal{G}[k] = (\mathcal{V}, \mathcal{E}[k])$ with $n  = | \mathcal{V} |$ nodes. 
Each node $v_i$ has a scalar quadratic local cost function $f_i : \mathbb{R}^n \mapsto \mathbb{R}$ (see \cite{2019:Yang_etal} and references therein) defined as: 
\begin{equation}\label{local_cost_functions}
    f_i(z) = \dfrac{1}{2} \alpha_i (z - \mu_i)^2 , 
\end{equation}
where $\alpha_i > 0$, $\mu_i \in \mathbb{R}$ is the demand of node $v_i$ and $z$ is a global optimization parameter which determines the data to be stored at each node. 
The global cost function is the sum of every local cost function $f_i$ (see \eqref{local_cost_functions}) in the network, i.e., 
\begin{equation}\label{global_cost_function}
    F(z) = \sum_{v_i \in \mathcal{V}} f_i(z) . 
\end{equation}
The main goal of the nodes is to distributively allocate the available data in order to calculate $z^*$ which minimizes the global cost function in \eqref{global_cost_function} and is defined as 
\begin{align}\label{opt:1}
z^* =  \arg\min_{z\in \mathcal{Z}} \sum_{v_{i} \in \mathcal{V}} f_i(z) , 
\end{align}
where $\mathcal{Z}$ is the set of feasible values of parameter $z$. 
Note that the solution $z^*$ in \eqref{opt:1} can be given in closed form as: 
\begin{align}\label{eq:closedform}
z^* =  \frac{\sum_{v_{i} \in \mathcal{V}} \alpha_i \mu_{i}}{\sum_{v_{i} \in \mathcal{V}} \alpha_i}.
\end{align}
Also, note here that if $\alpha_i =1$ for all $v_{i}\in\mathcal{V}$, the solution is the average.

The problem we present in this paper is borrowed from \cite{2020:Themis_Kalyvianaki, 2021:Rikos_TaskScheduling}, but is adjusted in the context of data allocation over dynamic wireless computer networks. 
Specifically, each node $v_j$ aims to calculate the optimal amount of data to receive $w_i^*[m]$ at each optimization step $m$, which fulfills 
\begin{align}\label{cond:balance}
\frac{w_i^*[m] +\delta_{i}[m]}{\nu_i^{\max}} &= \frac{w_j^*[m] +\delta_{j}[m]}{\nu_j^{\max}} \\
&= \frac{\mu[m] + \delta_{\mathrm{tot}}[m]}{\nu^{\max}}, \ \forall v_{i}, v_{j} \in \mathcal{V}. \nonumber
\end{align}
This means that every node aims to balance its data storage (i.e., the same percentage of stored data per available memory) during the algorithm's execution. 
In the remainder of this paper, we consider a single optimization step (i.e., without loss of generality, we drop index $m$). 
From \cite{2020:Themis_Kalyvianaki}, in order to fulfill \eqref{cond:balance}, we need: 
\begin{align}\label{eq:closedform1}
z^* =  \frac{\sum_{v_{i} \in \mathcal{V}} \nu_i^{\max} \frac{\mu_{i}+\delta_{i}}{\nu_i^{\max}}}{\sum_{v_{i} \in \mathcal{V}} \nu_i^{\max}} = \frac{\mu + \delta_{\mathrm{tot}}}{\nu^{\max}}.
\end{align}
Thus, the cost function $f_i(z)$ in \eqref{local_cost_functions} is given by
\begin{align}\label{eq:fiz}
f_i(z) = \frac{1}{2}\nu_i^{\max} \left(z- \frac{\mu_{i}+\delta_{i}}{\nu_i^{\max}} \right)^2.
\end{align}
This means that each node computes the optimal amount of data to store and then it is able to find the amount of data $w_i^*$ to receive, i.e.,
 \begin{align}\label{eq:optimal_workload}
w_i^*  = \frac{\mu + \delta_{\mathrm{tot}}}{\nu^{\max}} \nu_i^{\max} - \delta_{i}.
\end{align}

In our paper we aim to develop a distributed algorithm which operates over dynamic networks and during its operation, each node $v_j$ does the following:  
\begin{itemize}
    \item It calculates the optimal solution $w_j^*$ in \eqref{cond:balance} at every optimization step $m$. 
    \item It converges to the optimal solution after a finite number of time steps. 
    \item It processes and transmits quantized values. 
    \item It ceases transmissions once convergence has been achieved without having knowledge of global parameters (for preserving the available resources of each node). 
\end{itemize}

%
%
%
%
\section{Quantized Data Allocation Algorithm with Finite Transmission Capabilities}\label{sec:Dyn_Alloc_Alg_Finite_Tr}

In this section we present a distributed algorithm which solves the problem described in Section~\ref{sec:probForm}. 
The distributed algorithm is detailed below as Algorithm~\ref{algorithm_max} and allows each node in the network to calculate in finite time the optimal amount of data to receive. 
In order to solve the finite time data allocation problem, we make the following assumptions. 

\begin{assum}\label{str_digr}
Let us consider an infinite sequence of undirected graphs $\mathcal{G}[0], \mathcal{G}[1]$, $\mathcal{G}[2]$, ..., $\mathcal{G}[k]$, ..., describing a dynamic graph. 
There is a finite window length $l \in \mathbb{N}$ and an infinite sequence of time instants $t_0$, $t_1$, ..., $t_m$, ..., where $t_0 = 0$, such that for any $m \in \mathbb{Z}_+$, we have $0 < t_{m+1} - t_m < l < \infty$ and the union graph $\mathcal{G}^{t_{m}, ..., t_{m+1}-1}$, is equal to the nominal undirected graph $\mathcal{G}$ which is assumed to be connected. 
Furthermore, the diameter of the connected union graph $\mathcal{G}^{t_{m}, ..., t_{m+1}-1}$ is denoted as $D^{un}$ and is the longest shortest path between any two nodes $v_j, v_i \in \mathcal{V}$ (note that $D^{un}$ is also the diameter of the nominal graph $\mathcal{G}$). 
\end{assum}

\begin{assum}\label{ID_graph}
Each node $v_j$ has a unique ID. 
This ID is used to distinguish node $v_j$ from other nodes in the network. 
\end{assum}

\begin{assum}\label{assum_resource}
    %
    The time horizon $T_{o}$ at step $m$, is chosen such that $\mu[m] \leq \nu^{\mathrm{avail}}[m]$. 
    This means that the total amount of data to be allocated at a specific optimization step $m$ is smaller or equal to the total available memory of the network. 
\end{assum}

Assumption~\ref{str_digr} is a necessary network condition so that each node $v_j$ is be able to calculate the quantized average of each node's quantized state after a finite number of time steps. 
Assumption~\ref{ID_graph} is a necessary condition in order for each node to cease transmissions once the optimal allocation is calculated after a finite number of time steps. 
Assumption~\ref{assum_resource} is a necessary condition so that the total demand of data storage does not exceed the total available memory in the network. 
Note that $T_{o}$ can be chosen appropriately to fulfill this requirement. 
Furthermore, note that in case Assumption~\ref{assum_resource} does not hold, some data will not be allocated due to the lack of available storing memory in the system. 

We now describe the main operations of Algorithm~\ref{algorithm_max}. 
The initialization involves the following steps: 

\textbf{Initialization:} 
Each node $v_j \in \mathcal{V}$ does the following: (i) it initiates its mass variables, (ii) it initiates its transmission variables, and (iii) it sets its state variables to be equal to the mass variables. 
Then, it broadcasts the values of its state variables to every neighbor. 
Finally, it initializes set $S_j$, to contain the out-neighbors to which it transmitted its state variables at time step $k=0$. 

The iteration involves the following steps: 

\textbf{Iteration - Step~$1$. Probability Assignment and Receiving:} 
Each node $v_j \in \mathcal{V}$ assigns a nonzero probability to each of its edges during every time step $k$. 
The sum of the assigned nonzero probabilities is equal to one, during every time step $k$. 
Then, it receives from every neighbor (i) the transmitted set of state variables, and (ii) the transmitted set of mass variables (if no set of mass variables is received from a specific neighbor, $v_j$ assumes it received mass variables equal to zero from this neighbor). 

\textbf{Iteration - Step~$2$. Transmission Conditions According to Mass and State Variables:} 
If node $v_j$ received at least a set of mass variables or state variables during Iteration - Step~$1$, it checks the following conditions: 
\begin{list4}
    \item If the received set of state variables is ``greater'' (in the way clarified later in this section) than the current set of state variables, it sets its state variables to be equal to the received (greater) set of state variables and decides to broadcast its updated state variables. 
    \item If the stored set of mass variables is greater than the state variables it sets its state variables equal to the mass variables and decides to broadcast its updated state variables. 
    \item If the set of state variables is greater than the set of mass variables and the fraction of its mass variables is not equal to the fraction of its state variables, then it decides to directly transmit its mass variables to a randomly chosen neighbor. 
\end{list4} 
Then, if it decided to broadcast its state variables, it updates the stored set $S_j$ to be equal to the current set of out-neighbors (in order to remember which neighbors received the current state variables). 

\textbf{Iteration - Step~$3$. Transmission Conditions According to Dynamic Network:} 
At each time step $k$, each node $v_j$ checks whether the current set of neighbors is not included in the stored set $S_j$. 
Note here that the stored set $S_j$ denotes the neighbors which have received (or will receive) the current set of state variables. 
In case, there is one (or multiple) neighbor(s) who is (are) not included in the set $S_j$, node $v_j$ decides to broadcast its state variables so that this one neighbor (or multiple neighbors) receives the updated set of state variables. 

\textbf{Iteration - Step~$4$. Transmitting:} 
Each node $v_j$ checks its transmission variables. 
It transmits its mass variables via a direct transmission, or broadcasts its state variables. 
Then, it sets its tramsmission variables equal to zero and repeats the operation.

The details of the dynamic algorithm with transmission stopping capabilities can be seen in Algorithm~\ref{algorithm_max}. 

\noindent
\vspace{-0.5cm}    
\begin{varalgorithm}{1}
\caption{Optimal Data Allocation Algorithm with Efficient Communication and Transmission Stopping}
\textbf{Input:} A set of graphs $\mathcal{G}[k] = (\mathcal{V}, \mathcal{E}[k])$ with $n=|\mathcal{V}|$ nodes and $m[k]=|\mathcal{E}[k]|$ edges for which Assumption~\ref{str_digr} holds. 
\\
\textbf{Initialization:} Each node $v_j \in \mathcal{V}$ does the following: 
\begin{list4}
\item[1)] Sets $z_j[0] := l_j + \delta_j$, $y_j[0] = \nu_j^{\max}$, $z^s_j[0] = 1$, $y^s_j[0] = y_j[0]$, $q^s_j[0] = y^s_j[0] / z^s_j[0]$ and $S\_br_j = 0$, $M\_tr_j = 0$. 
\item[2)] Broadcasts $z^s_j[0]$, $y^s_j[0]$ to every $v_l \in \mathcal{N}_j$.
\item[3)] Sets $S = \mathcal{N}_j[0]$. 
\end{list4}
\textbf{Iteration:} For $k=0,1,2,\dots$, each node $v_j \in \mathcal{V}$ does the following: 
\begin{list4}
\item[1)] Assigns a nonzero probability $b_{lj}[k]$ to each of its edges $m_{lj}$, where $v_l \in \mathcal{N}_j[k]$, as follows 
\begin{align*}
b_{lj}[k] = \left\{ \begin{array}{ll}
         \frac{1}{\mathcal{D}_j[k]+1}, & \mbox{if $l = j$ or $v_{l} \in \mathcal{N}_j[k]$,}\\
         0, & \mbox{if $l \neq j$ and $v_{l} \notin \mathcal{N}_j[k]$.}\end{array} \right. 
\end{align*}
\item[2)] Receives $y^s_i[k]$, $z^s_i[k]$ from every $v_i \in \mathcal{N}_j[k]$ (if no message is received, it sets $y^s_i[k] = 0$, $z^s_i[k] = 0$). 
\item[3)] Receives $y_i[k]$, $z_i[k]$ from each $v_i \in \mathcal{N}_j[k]$ and sets 
$$
y_j[k+1] = y_j[k] + \sum_{v_i \in \mathcal{N}_j[k]} w_{ji}[k]y_i[k] ,
$$  \vspace{-.3cm}
$$
z_j[k+1] = z_j[k] + \sum_{v_i \in \mathcal{N}_j[k]} w_{ji}[k]z_i[k] ,
$$
where $w_{ji}[k]=1$ if a message with $y_i[k]$, $z_i[k]$ is received from in-neighbor $v_i$, otherwise $w_{ji}[k]=0$. 
\item[4)] \textbf{If} $w_{ji}[k] \neq 0$ or $z^s_i[k] \neq 0$ for some $v_i \in \mathcal{N}_j^-[k]$ \textbf{then}
\begin{list4}
\item[4a)] Calls Algorithm~\ref{algorithm_max_1a}. 
\item[4b)] \underline{Event Trigger Conditions~$1$:} \textbf{If} 
$
S_j \ \cap \ \mathcal{N}_j^+[k]  \neq  \emptyset ,
$
\\ \textbf{then} node $v_j$ sets $S\_br_j = 1$, and $S_j = S_j \cup \mathcal{N}_j[k]$. 
\item[4c)] \textbf{If} $M\_tr_j = 1$ \textbf{then} node $v_j$ chooses $v_l \in \mathcal{N}_j[k]$ randomly according to $b_{lj}[k]$ and transmits $y_j[k]$, $z_j[k]$. 
Then, node $v_j$ sets $y_j[k] = 0$, $z_j[k] = 0$, $M\_tr_j = 0$. 
\item[4d)] \textbf{If} $S\_br_j = 1$ \textbf{then}, node $v_j$ broadcasts $z^s_j[k+1]$, $y^s_j[k+1]$ to every $v_l \in \mathcal{N}_j$. 
Then, it sets $S\_br_j = 0$. 
\end{list4}
\item[5)] Repeats (increases $k$ to $k + 1$ and goes back to Step~$1$).
\end{list4}
\textbf{Output:} Sets $w^*_j + \delta_{j} = (\nu_j^{\max} / q^s_j[k])$ and \eqref{cond:balance} holds for every $v_j \in \mathcal{V}$. 
\label{algorithm_max}
\end{varalgorithm}

\noindent
\vspace{-0.5cm}    
\begin{varalgorithm}{1.A}
\caption{Event-Triggered Conditions for Algorithm~\ref{algorithm_max} (for each node $v_j$)}
\textbf{Input} 
\\ $y^s_j[k]$, $z^s_j[k]$, $q^s_j[k]$, $y_j[k+1]$, $z_j[k+1]$, $S\_br_j$, $M\_tr_j$, $S_j$, $\mathcal{N}_j[k]$, and the received $y^s_i[k]$, $z^s_i[k]$ from every $v_i \in \mathcal{N}_j[k]$.
\\
\textbf{Execution} 
\begin{list4}
\item \underline{Event Trigger Conditions~$1$:} \textbf{If} 
\\ Condition~$(i)$: $z^s_i[k] > z^s_j[k]$, or
\\ Condition~$(ii)$: $z^s_i[k] = z^s_j[k]$ and $y^s_i[k] > y^s_j[k]$, 
\\ \textbf{then} node $v_j$ sets 
$$ 
z^s_j[k+1] = \max_{v_i \in \mathcal{N}_j[k]} z^s_i[k] , \ \ \text{and}
$$ 
$$ 
y^s_j[k+1] = \max_{v_i \in \{v_{i'} \in \mathcal{N}_j[k] | z^s_{i'}[k] = z^s_j[k+1]\}} y^s_i[k] ,
$$ 
and also sets $q^s_j[k+1] = \frac{y^s_j[k+1]}{z^s_j[k+1]}$, and $S\_br_j = 1$. 
\item \underline{Event Trigger Conditions~$2$:} \textbf{If}
\\ Condition~$(i)$: $z_j[k+1] > z^s_j[k+1]$, or 
\\ Condition~$(ii)$: $z_j[k+1] = z^s_j[k+1]$ and $y_j[k+1] > y^s_j[k+1]$, 
\\ \textbf{then} node $v_j$ sets $z^s_j[k+1] = z_j[k+1]$, $y^s_j[k+1] = y_j[k+1]$, and $q^s_j[k+1] = \frac{y^s_j[k+1]}{z^s_j[k+1]}$ and $S\_br_j = 1$. 
\item \underline{Event Trigger Conditions~$3$:} \textbf{If}
\\ Condition~$(i)$: $0 < z_j[k+1] < z^s_j[k+1]$ or 
\\ Condition~$(ii)$: $z_j[k+1] = z^s_j[k+1]$ and $y_j[k+1] < y^s_j[k+1]$, 
\\ \textbf{then} node $v_j$ sets $M\_tr_j = 1$.
\item \underline{Event Trigger Conditions~$4$:} \textbf{If} $0 < z_j[k+1]$ and 
$$
\dfrac{y_j[k+1]}{z_j[k+1]} = \dfrac{y^s_j[k+1]}{z^s_j[k+1]} , 
$$
\textbf{then} node $v_j$ sets $M\_tr_j = 0$.
\item \underline{Event Trigger Conditions~$5$:} \textbf{If} $S\_br_j = 1$
\\ \textbf{then} node $v_j$ sets $S = \mathcal{N}_j[k]$. 
\end{list4} 
\textbf{Output} 
\\ $y^s_j[k]$, $z^s_j[k]$, $q^s_j[k]$, $S\_br_j$, $M\_tr_j$, $S_j$. 
\label{algorithm_max_1a}
\end{varalgorithm}


\vspace{.1cm}

\begin{remark}
In Definition~\ref{ID_graph}, each node $v_j$ has a unique ID in order to cease transmissions once the optimal allocation is calculated after a finite number of time steps. 
To the authors knowledge this is the first algorithm which operates over dynamic networks in which nodes are able to cease transmissions without any global parameter such as the network diameter. 
However, note that if nodes have knowledge of the parameter $l$ (see Definition~\ref{str_digr}), then they do not need to have unique IDs. 
Specifically, every time their state variables are updated (see Algorithm~\ref{algorithm_max_1a}) they can broadcast their state variables for $l$ time steps. 
From Definition~\ref{str_digr}, every $l$ time steps, there is a link from node $v_j$ to every neighboring node. 
As a result, if nodes broadcast their state variables for $l$ time steps, every neighbor will receive the updated states at least once and the algorithm will converge to the optimal solution in finite time. 
\end{remark}

\subsection{Operation over Dynamic Graphs}\label{dynamic_oper_analysis}

We now analyze the functionality of Algorithm~\ref{algorithm_max} over dynamic networks. 
We consider the following two definitions which are important for our subsequent development. 

\begin{defn}\label{leadingmass_defin}
Consider a set of graphs $\mathcal{G}[k] = (\mathcal{V}, \mathcal{E}[k])$, $k=0, 1, 2, ...$, with $n=|\mathcal{V}|$ nodes and $m[k]=|\mathcal{E}[k]|$ edges for which Assumption~\ref{str_digr} holds. 
During the execution of Algorithm~\ref{algorithm_max}, at time step $k_0$, there is at least one node $v_{j'} \in \mathcal{V}$, for which 
\begin{equation}\label{great_z_prop1_det}
z_{j'}[k_0] \geq z_i[k_0], \ \forall v_i \in \mathcal{V}.
\end{equation}
Then, among the nodes $v_{j'}$ for which (\ref{great_z_prop1_det}) holds, there is at least one node $v_j$ for which 
\begin{equation}\label{great_z_prop2_det}
y_j[k_0] \geq y_{l}[k_0] , \ \text{where} \ \ v_j, v_{l} \in \{ v_{j'} \in \mathcal{V} \ | \ (\ref{great_z_prop1_det}) \ \text{holds} \}.
\end{equation}
For notational convenience we will call the pair of mass variables of node $v_j$ for which (\ref{great_z_prop1_det}) and (\ref{great_z_prop2_det}) hold as the ``leading mass'' (or ``leading masses'' if multiple nodes hold such a pair of values) and the pairs of mass variables of a node $v_l$ for which $z_l[k_0]>0$ but (\ref{great_z_prop1_det}) and (\ref{great_z_prop2_det}) do not hold as the ``follower mass'' (or ``follower masses'').
\end{defn}

\begin{defn}\label{mass_merge}
Consider a set of graphs $\mathcal{G}[k] = (\mathcal{V}, \mathcal{E}[k])$, $k=0, 1, 2, ...$, with $n=|\mathcal{V}|$ nodes and $m[k]=|\mathcal{E}[k]|$ edges for which Assumption~\ref{str_digr} holds. 
During the execution of Algorithm~\ref{algorithm_max}, at time step $k_0$, if two (or more) masses (for which $z \neq 0$) reach a node simultaneously then we say that they ``merge''. 
This means that the receiving node ``merges'' the mass variables it receives by summing their numerators and their denominators (according to Step~$3$ of the Iteration of Algorithm~\ref{algorithm_max}). 
This way a set of mass variables with a greater denominator is created.
\end{defn}

The intuition behind Algorithm~\ref{algorithm_max} can be described through the following three stages. 
\\ \underline{Stage~$1$:} Initially every node assumes that its mass variables are the leading mass and broadcasts its state variables. 
After a finite number of time steps, the state variables of every node in the network are equal to the leading mass (let us assume the simple scenario where until the state variables of every node in the network are equal to the leading mass no mass variables merged, thus the leading mass do not change). 
\\ \underline{Stage~$2$:} Once the state variables of every node become equal to the leading mass, every node transmits its mass variables towards a randomly chosen neighbor. 
This means that the mass variables of every node (except the node whose mass variables are the leading mass) perform a random walk. 
During their random walk, the mass variables either merge with the leading mass (which is not transmitted), or merge between them (if they visit a common node). 
In the first case, the leading mass is updated and the corresponding node broadcasts its updated state variables. 
In the second case, if two mass variables visit a common node, the node checks if the merged mass variables are now the leading mass (note that the node's state variables are equal to the leading mass from Stage~$1$). 
If the merged mass variables are the leading mass, then the corresponding node broadcasts its state variables and stores the mass variables (i.e., it does not transmit its mass variables). 
However, if the merged mass is not the leading mass, then it is transmitted to a randomly chosen neighbor. 
\\ \underline{Stage~$3$:} Once the leading mass is updated (i.e., the mass variables either merge with the leading mass, or merge between them), the corresponding node broadcasts its state variables. 
Thus, after a finite number of time steps, the state variables of every node in the network become equal to the updated leading mass, and then Stage~$2$ is repeated.

Note here that during the operation of Algorithm~\ref{algorithm_max}, there is always a set of mass variables which is the leading mass. 
The follower masses perform random walks until they merge with the leading mass or they merge between them and become the leading mass. 
Once the leading mass is updated (i.e., either a follower mass merges with the leading mass or two follower masses merge between them and become the leading mass), every node in the network receives the updated state variables of the node whose mass variables are the leading mass (see Stage~$1$). 
As a result, after a finite number of time steps (i) the leading mass becomes equal to the quantized average of the initial states, (ii) each node sets its state variables equal to the leading mass (i.e., the average of the initial states), and (iii) once each node's state becomes equal to the average of the initial states, transmissions are ceased (since there are no more updates of the leading mass). 

\vspace{.3cm}


\textbf{Comparison with Previous Works.}
It is important to note that Algorithm~\ref{algorithm_max} is significantly different from the optimal allocation algorithm in \cite{2021:Rikos_TaskScheduling} and the quantized average consensus algorithm in \cite{2021:Rikos_Hadj_Fin_Tr_Av_Con}. 
In \cite{2021:Rikos_TaskScheduling} the authors present a distributed algorithm for optimally allocating resources over a directed communication network. 
The algorithm converges in finite time and exhibits a fast convergence rate over large scale networks. 
However, (i) it operates over static directed networks, (ii) its transmission stopping mechanism requires knowledge of the diameter of the network (which is a global parameter), and (iii) it introduces a quantization error (which is upper bounded by the size of the quantization step) to the final calculated value, due to its operation, which relies on communication and processing of quantized values. 
In \cite{2021:Rikos_Hadj_Fin_Tr_Av_Con} the authors propose a distributed algorithm for calculating the average of the initial states over a directed communication network. 
This algorithm achieves deterministic convergence and requires a finite number of time steps upper bounded by a polynomial function. 
However, its deterministic finite time convergence requires (i) the existence of a static directed communication network (i.e., it is not adjusted to operate over time-varying networks), and (iii) a large number of time steps, due to its round-robin transmission strategy. 
Algorithm~\ref{algorithm_max} operates over dynamic undirected networks. 
Furthermore, it incorporates a transmission stopping mechanism, which (i) does not require knowledge of any global parameter of the network, and (ii) is adjusted to the dynamic nature of the network (note here that it is the fist transmission stopping mechanism which operates over time-varying networks). 
Additionally, Algorithm~\ref{algorithm_max} convergences almost surely (with probability arbitrarily close to one) to the optimal solution due to the randomized nature of the directed transmissions. 
Finally, Algorithm~\ref{algorithm_max} converges to the optimal solution in finite time without introducing a quantization error (i.e., the exact optimal allocation value is calculated in the form of a quantized fraction). 

\subsection{Convergence of Algorithm~\ref{algorithm_max}}\label{ConvDynamicAlg}

We now analyze the convergence time of Algorithm~\ref{algorithm_max}. 
We first consider Lemma~\ref{onetokenwalk}, \textit{mutatis mutandis}, which is necessary for our subsequent development. 

\begin{lemma}[\hspace{-0.00001cm}\cite{2020:RikosHadj_IFAC}]
\label{onetokenwalk}
Consider a sequence of graphs $\mathcal{G}[k] = (\mathcal{V}, \mathcal{E}[k])$, $k=0, 1, 2, ...$, with $n=|\mathcal{V}|$ nodes, $m[k]=|\mathcal{E}[k]|$ edges, so that Assumption~\ref{str_digr} holds for $\mathcal{G}[k]$ over all $k$. 
At each time step $k$, suppose that each node $v_j$ assigns a nonzero probability $b_{lj}[k]$ to each of its edges $m_{lj}[k]$, where $v_l \in \mathcal{N}_j[k] \cup \{v_j\}$, as follows  
\begin{align*}
b_{lj} = \left\{ \begin{array}{ll}
         \frac{1}{1 + \mathcal{D}_j[k]}, & \mbox{if $l = j$ or $v_{l} \in \mathcal{N}_j[k]$,}\\
         0, & \mbox{if $l \neq j$ and $v_{l} \notin \mathcal{N}_j[k]$.}\end{array} \right. 
\end{align*}
At time step $k=0$, node $v_j$ holds a ``token" while the other nodes $v_l \in \mathcal{V} - \{ v_j \}$ do not. 
At each time step $k$, each node $v_j$ transmits the ``token'' (if it has the token, otherwise it performs no transmission) according to the nonzero probability $b_{lj}[k]$ it assigned to its edges $m_{lj}[k]$. 
The probability $P^{D^{un}}_{DT_i}$ that the token is at node $v_i$ after $l D^{un}$ time steps satisfies 
$$
P^{l D^{un}}_{DT_i} \geq (1+\mathcal{D}_{max})^{-(l D^{un})} > 0 , 
$$
where $l$ is the time window defined in Assumption~\ref{str_digr} (for which the union graph $\mathcal{G}^{t_{m}, ..., t_{m+1}-1}_d$ is equal to the nominal graph $\mathcal{G}$ which is connected), and $\mathcal{D}_{max}$ is the maximum degree of every node in the nominal graph $\mathcal{G}$. 
\end{lemma}

We now consider Lemma~\ref{tokenvisit}, which analyzes the probability according to which a token performing a random walk visits a specific node. 

\begin{lemma}
\label{tokenvisit}
Consider a sequence of graphs $\mathcal{G}[k] = (\mathcal{V}, \mathcal{E}[k])$, $k=0, 1, 2, ...$, with $n=|\mathcal{V}|$ nodes, $m[k]=|\mathcal{E}[k]|$ edges, so that Assumption~\ref{str_digr} holds for $\mathcal{G}[k]$ over all $k$. 
At each time step $k$, suppose that each node $v_j$ assigns a nonzero probability $b_{lj}[k]$ to each of its edges $m_{lj}[k]$, where $v_l \in \mathcal{N}_j[k] \cup \{v_j\}$, as follows  
\begin{align*}
b_{lj} = \left\{ \begin{array}{ll}
         \frac{1}{1 + \mathcal{D}_j[k]}, & \mbox{if $l = j$ or $v_{l} \in \mathcal{N}_j[k]$,}\\
         0, & \mbox{if $l \neq j$ and $v_{l} \notin \mathcal{N}_j[k]$.}\end{array} \right. 
\end{align*}
At time step $k=0$, node $v_j$ holds a ``token" while the other nodes $v_l \in \mathcal{V} - \{ v_j \}$ do not. 
At each time step $k$, each node $v_j$ transmits the ``token'' (if it has the token, otherwise it performs no transmission) according to the nonzero probability $b_{lj}[k]$ it assigned to its edges $m_{lj}[k]$. 
For any probability $p_0$, where $0 < p_0 < 1$, there exists $k_0 \in \mathbb{Z}_+$, so that with probability at least $p_0$, the token has visited a specific node $v_i$, (where $l$ is the time window defined in Assumption~\ref{str_digr} for which the union graph $\mathcal{G}^{t_{m}, ..., t_{m+1}-1}$ is equal to the nominal graph $\mathcal{G}$ which is connected). 
\end{lemma}

\begin{proof}
See Appendix~\ref{appendix:A}.
\end{proof}

We are now ready to present Theorem~\ref{converge_Alg1} which analyzes the finite time convergence of Algorithm~\ref{algorithm_max}. 

\begin{theorem}\label{converge_Alg1}
Consider a sequence of graphs $\mathcal{G}[k] = (\mathcal{V}, \mathcal{E}[k])$, $k=0, 1, 2, ...$, with $n=|\mathcal{V}|$ nodes, $m[k]=|\mathcal{E}[k]|$ edges, so that Assumption~\ref{str_digr}, Assumption~\ref{ID_graph}, and Assumption~\ref{assum_resource} hold for $\mathcal{G}[k]$ over all $k$. 
Suppose that each node $v_j \in \mathcal{V}$ follows the Initialization and Iteration steps as described in Algorithm~\ref{algorithm_max}, where $l_j, \delta_j, \nu_j^{\max} \in \mathbb{N}$ for every node $v_j \in \mathcal{V}$ at time step $k=0$. 
During the operation of Algorithm~\ref{algorithm_max}, for any probability $p_0'$ (where $0 < p_0' < 1$) there exists $k_0' \in \mathbb{Z}_+$, so that with probability at least $p_0'$ each node $v_j$ is able to (i) calculate the optimal amount of data $w^*_j$ (shown in \eqref{eq:optimal_workload}) after a finite number of time steps $k_0$, and (ii) cease transmissions after calculating $w^*_j$. 
\end{theorem}

\begin{proof}
See Appendix~\ref{appendix:B}.
\end{proof}


%
%
%
%

\section{Simulation Results} \label{sec:results}

In this section, we present simulation results in order to demonstrate the operation of Algorithm~\ref{algorithm_max} and its potential advantages. 
We focus on a random graph of $20$ nodes and show how the nodes' states converge to the optimal solution. 
Furthermore, we show the total accumulated number of transmissions and the number transmissions at every time step. 
To the best of our knowledge, this is the first work that faces the problem of optimal resource allocation using quantized values with transmission stopping guarantees over dynamic networks. 

\textbf{Evaluation over a Dynamic Network of $20$ Nodes.} 
The dynamic network is comprised of $20$ nodes and the union of the dynamic networks is equal to the nominal graph after $l = 5$ time steps. 
The nominal graph is assumed to be connected and has a diameter equal to $3$. 
At each node $v_j$, the total load of data $l_j$ was generated via a random distribution uniformly picked within the range $[1, 50]$. 
The total load of data in the network is equal to $504$ (i.e., $\sum_{v_j \in \mathcal{V}} l_j = 504$). 
For a randomly chosen set of seven nodes the total memory was set to be $31752$, for a randomly chosen set of seven nodes the total memory was set to be $63504$, and for a randomly chosen set of six nodes the total memory was set to be $95256$. 
Our simulation results are shown in Fig.~\ref{optimality} and Fig.~\ref{transmissions}. 

In Fig.~\ref{optimality} (A), we can see that each node $v_j$ is able to calculate the exact ratio of memory per data after $29$ time steps. 
The ratio is equal to $504 / 20$ and and is calculated exactly in the form of a quantized fraction without any errors due to quantized communication and processing. 
In Fig.~\ref{optimality} (B), we can see that each node $v_j$ is able to calculate the optimal amount of data to receive after $29$ time steps. 
The amount of data is proportional to the node's memory capacity. 
Specifically, from the result of Fig.~\ref{optimality}, each node calculates the ratio of data per memory and then scales with its available memory capacity. 
Specifically, the $7$ nodes with total memory equal to $31752$, receive $1260$ amount of data. 
The $7$ nodes with total memory equal to $63504$, receive $2520$ amount of data (double the amount received by the nodes with $31752$). 
Finally, the $6$ nodes with total memory equal to $95256$, receive $3780$ amount of data. 

In Fig.~\ref{transmissions} (A), we can see the accumulated total number of transmissions performed from nodes in the network during the operation of Algorithm~\ref{algorithm_max}. 
The total number of transmissions performed is equal to $291$ during $29$ time steps. 
In Fig.~\ref{transmissions} (B), we can see the number of transmissions during the operation of Algorithm~\ref{algorithm_max} at every time step $k$. 
We can see, that in the beginning, the number of transmissions is high. 
However, after $10$ time steps, it is decreased and it increases at specific instances due to the dynamic nature of the communication network. 
The number of transmissions performed becomes equal to $0$ after $29$ time steps. 

\begin{figure}[t]
\begin{center}
\includegraphics[width=.9\columnwidth]{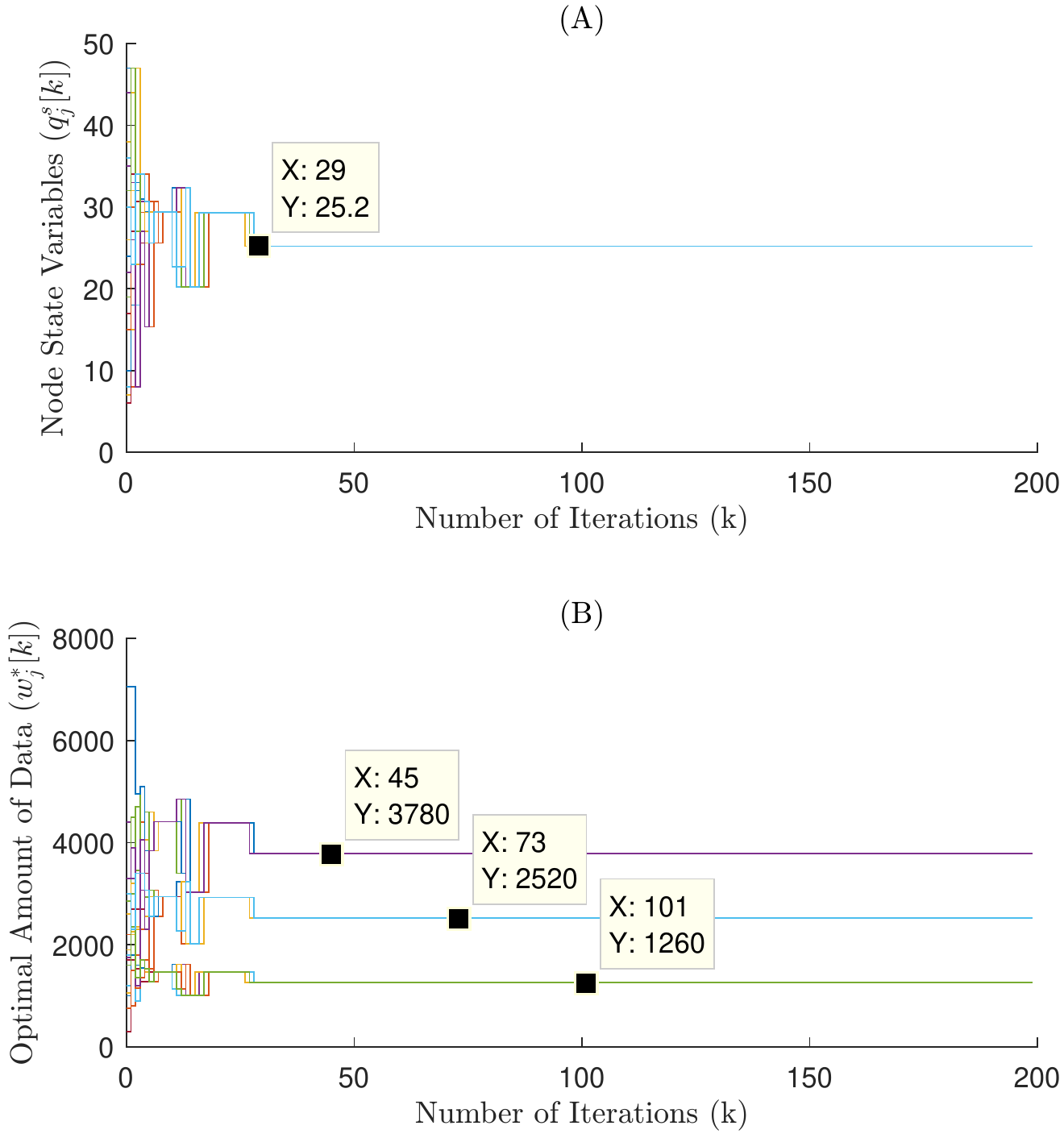}
\caption{Execution of Algorithm~\ref{algorithm_max} over a random dynamic graph comprised of $20$ nodes where the union of the dynamic graphs is equal to the nominal graph after $l = 5$ time steps. 
\textit{(A):} During Algorithm~\ref{algorithm_max} every node calculates the amount of data per memory after $29$ time steps. 
\textit{(B):} Algorithm~\ref{algorithm_max} converges to the optimal solution after $29$ time steps\vspace{-0.2cm}.}
\label{optimality}
\end{center}
\end{figure}

\begin{figure}[t]
\begin{center}
\includegraphics[width=.85\columnwidth]{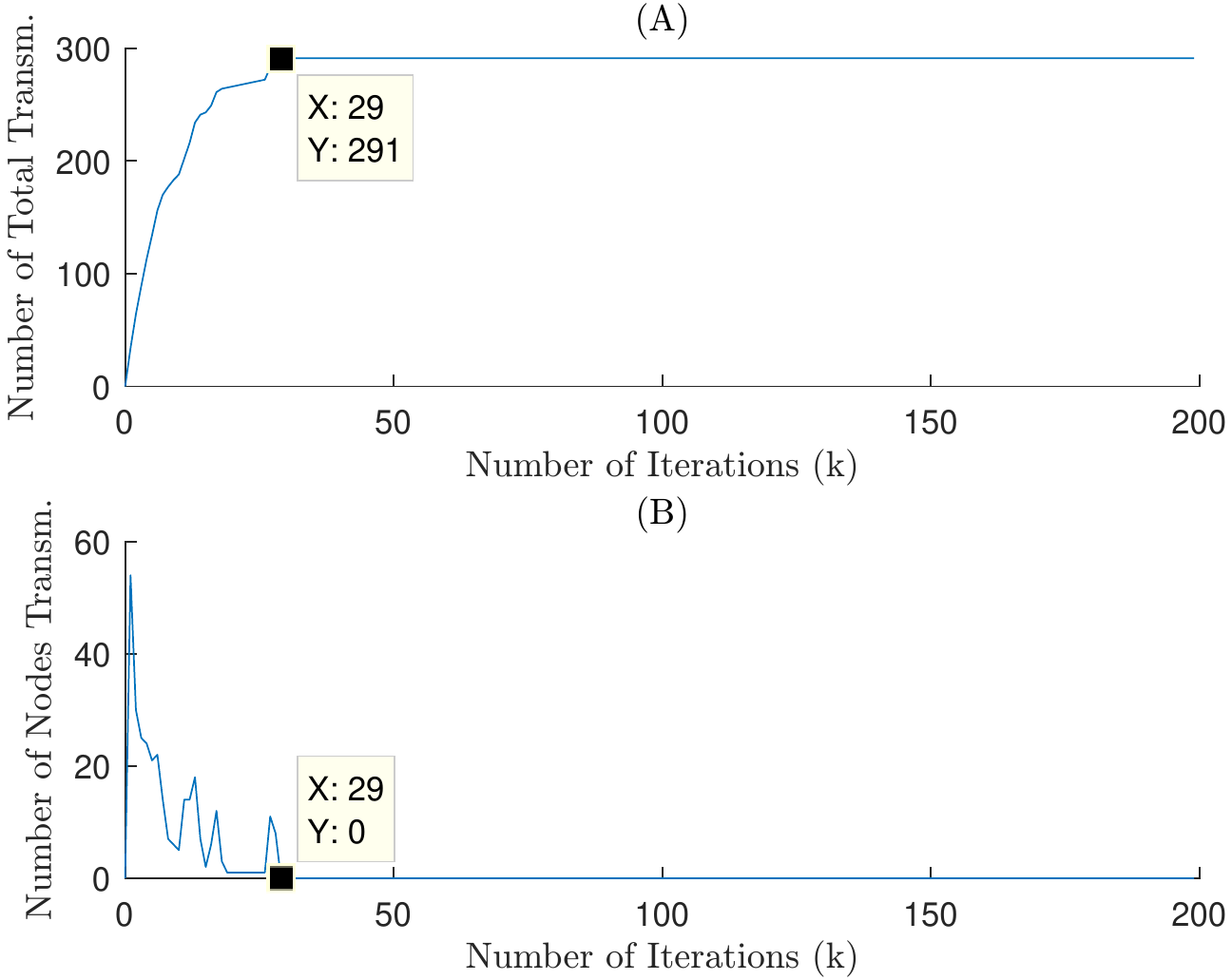}
\caption{Execution of Algorithm~\ref{algorithm_max} over a random dynamic graph comprised of $20$ nodes where the union of the dynamic graphs is equal to the nominal graph after $l = 5$ time steps. 
\textit{(A):} During Algorithm~\ref{algorithm_max} the total number of transmissions performed is equal to $291$. \textit{(B):} During Algorithm~\ref{algorithm_max} every node ceases transmissions after $29$ time steps\vspace{-0.2cm}.}
\label{transmissions}
\end{center}
\end{figure}

%
%
%
%

\section{Conclusions and Future Directions}\label{sec:conclusions}

In this paper, we focused on the problem of optimal data scheduling over a wireless computing network. 
We proposed a distributed algorithm which operates over dynamic networks and converges in finite time. 
We showed that our algorithm converges to the exact optimal solution in finite time. 
Our algorithm operates with quantized values (i.e., each node processes and transmits quantized values) and, once it converges to the optimal solution, each node ceases transmissions. 
Finally, we have demonstrated the operation of our algorithm over random dynamic networks and exhorted its finite time convergence. 

Our proposed algorithm operates over dynamic undirected networks, thus its operation should be extended over dynamic directed networks. 
Also, designing a privacy preserving protocol is important in order to protect the privacy of the node which initially stores the largest amount of data in the network. 


\appendices
%
%
%
%
\section{Proof of Lemma~\ref{tokenvisit}} 
\label{appendix:A} 


The token performs random walk over a dynamic graph $\mathcal{G}[k]$. 
Since $b_{lj}[k] \geq (1+\mathcal{D}_{max})^{-1}$ (where $\mathcal{D}_{max}$ is defined in Lemma~\ref{onetokenwalk}) and Assumption~\ref{str_digr} holds for $\mathcal{G}[k]$ during all $k$, we have that the probability $P^{l D^{un}}_{DT^{out}}$ that ``the token is at node $v_i$ after $l D^{un}$ time steps'' is 
\begin{equation}\label{lowerProf_no_oscil_dyn}
P^{l D^{un}}_{DT^{out}} \geq (1+\mathcal{D}_{max})^{-(l D^{un})} > 0 .
\end{equation}
This is mainly due to the fact that every $l$ time steps, each edge is active for at least one time step. 
Since the nominal graph $\mathcal{G}_d$ is connected, it has a path of length at most $D^{un}$ from each node $v_{\ell}$ to each node $v_i$. 
Thus, at the first $l$ steps, we can select the first edge in this path (at the instant when it is active) and use self loops at the remaining instants; during the next $l$ time steps, we can select the second edge on this path and use self loops at the remaining instants; and so forth. 
From \eqref{lowerProf_no_oscil_dyn} we have that the probability $P^{l D^{un}}_{N\_ DT^{out}}$ that ``the token is not at node $v_i$ after $l D^{un}$ time steps'' is 
$
P^{lD^{un}}_{N\_ DT^{out}} \leq 1 - (1+\mathcal{D}_{max})^{-(l D^{un})} .
$
By extending this analysis, we choose $\varepsilon$ (where $0 < \varepsilon < 1$) for which it holds that 
\begin{equation}\label{varepsilon_condition}
\varepsilon \leq 1 - 2^{\log_2{p_0}} .
\end{equation}
We can state that for $\varepsilon$ which fulfills \eqref{varepsilon_condition} and after $\tau l D^{un}$ time steps where
\begin{equation}\label{windows_for_conv_1_no_oscil}
\tau \geq \Big \lceil \dfrac{\log{\varepsilon}}{\log{(1 - (1+\mathcal{D}^+_{max})^{-(l D^{un})})}} \Big \rceil ,
\end{equation}
the probability $P^{\tau}_{N\_ DT^{out}}$ that ``the token has not visited node $v_i$ after $\tau l D^{un}$ time steps'' is
$
P^{\tau}_{N\_ DT^{out}} \leq [P^{l D^{un}}_{N\_ T^{out}}]^{\tau} \leq \varepsilon .
$
As a result, after $\tau l D^{un}$ time steps (where $\tau$ fulfills \eqref{windows_for_conv_1_no_oscil}) the probability that ``the specific token $T_{\lambda}^{out}$ has visited node $v_i$ after $\tau l D^{un}$ time steps'' is equal to $1 - \varepsilon > p_0$, where $0 < p_0 < 1$, and $\varepsilon$ fulfills $\eqref{varepsilon_condition}$.

\section{Proof of Theorem~\ref{converge_Alg1}} 
\label{appendix:B} 

The intuition of the proof is the following. 
During the operation of Algorithm~\ref{algorithm_max}, let us consider the simple scenario where among the initial nodes' states there is only one ``leading mass'' and $n-1$ ``follower masses'' (see Definition~\ref{leadingmass_defin}). 
Initially, each node broadcasts its state variables. 
From Algorithm~\ref{algorithm_max_1a}, 
after $l D^{un}$ time steps the state variables of every node become equal to the leading mass (we assume that until the state variables of every node in the network are equal to the leading mass no mass variables merged, thus the leading mass did not change). 
Then, the follower masses perform random walk in the network. 
From Lemma~\ref{tokenvisit}, after $\tau l D^{un}$ time steps, where $\tau$ fulfills \eqref{windows_for_conv_1_no_oscil}, at least one follower mass will merge with the leading mass with high probability. 
Then, the node whose mass variables are the updated leading mass broadcasts its updated state variables.
Thus, the previous process is repeated until all the $n-1$ follower masses merge with the leading mass. 

During the operation of Algorithm~\ref{algorithm_max}, let us consider that there is only one ``leading mass'' and $n-1$ ``follower masses'' (the scenario of multiple leading masses can be proven identically). 
From Lemma~\ref{tokenvisit}, after $l D^{un} + \tau l D^{un}$ time steps, at least one follower mass will merge with the leading mass with probability at least $1 - \varepsilon$, where $\varepsilon$ is defined in $\eqref{varepsilon_condition}$. 
This means that after $(n-1) (l D^{un}) + (n-1) (\tau l D^{un})$ time steps, all $n-1$ follower masses will merge with the leading mass with probability $(1 - \varepsilon)^{n-1}$.
Choosing $\varepsilon$ to fulfill 
\begin{equation}\label{varepsilon_condition2}
\varepsilon \leq 1 - 2^{\frac{\log_2{p_0'}}{n-1}} , 
\end{equation}
and $\tau$ to fulfill \eqref{windows_for_conv_1_no_oscil}, we have that after $(n-1) (l D^{un}) + (n-1) (\tau l D^{un})$ time steps, all $n-1$ follower masses will merge with the leading mass with probability $(1 - \varepsilon)^{n-1} \geq p_0'$. 

As a result, during the operation of Algorithm~\ref{algorithm_max}, after $(n) (l D^{un}) + (n-1) (\tau l D^{un})$ time steps, each node $v_j$ is able to (i) calculate the optimal amount of data $w^*_j$ (shown in \eqref{eq:optimal_workload}) with probability $(1 - \varepsilon)^{n-1} \geq p_0'$, and (ii) cease transmissions after calculating $w^*_j$. 
Note here that the cases where (i) there are initially multiple leading masses, or (ii) follower masses perform random walk and two or more merge in order to form a new leading mass, can be proven identically. 

\vspace{-.3cm}



\bibliographystyle{IEEEtran}
\bibliography{bibliografia_consensus}

\end{document}